\documentclass[a4paper, number, sort&compress]{elsarticle}

\usepackage{hyperref}
\usepackage{microtype}
\usepackage{graphicx}
\usepackage{bbm}
\usepackage{enumerate}
\usepackage{booktabs}
\usepackage{tabu}
\usepackage[ruled,lined,boxed,linesnumbered]{algorithm2e}
\usepackage{amsfonts,amsthm,amsmath,amssymb,mathtools}
\usepackage{xcolor}
\usepackage{xspace}
\usepackage{float}
\usepackage{mathtools}
\usepackage{tabularx}
\usepackage{relsize}
\usepackage[czech,english]{babel}
\usepackage{todonotes}

\newtheorem{Reduction Rule}{Reduction Rule}

\newtheorem{theorem}{Theorem}
\newtheorem{lemma}{Lemma}
\newtheorem{corollary}{Corollary}
\newtheorem{claim}{Claim}

\usepackage{microtype}
\usepackage{xspace}

\newcommand{\stpath}{{$s$-$t$ path}\xspace}
\newcommand{\stpaths}{{$s$-$t$ paths}\xspace}

\newcommand{\sstpaths}{{shortest $s$-$t$ paths}\xspace}

\newcommand{\tsp}{{\sc Tracking Shortest Paths}\xspace}
\newcommand{\tpdag}{{\sc Tracking Paths in DAG}\xspace}
\newcommand{\tp}{{\sc Tracking Paths}\xspace}

\newcommand{\N}{\mathbb{N}}

\newcommand{\decprob}[3]{
  \vspace{2mm}
\noindent\fbox{
  \begin{minipage}{0.96\textwidth}
  \textsc{#1}\\
    \textbf{Input:} #2 \\
  \textbf{Question:} #3
  \end{minipage}
  }
  \vspace{2mm}
}

\newcommand\putabove[2]{\mathrel{\overset{\makebox[0pt]{\mbox{\normalfont\tiny\sffamily #1}}}{#2}}}

\usepackage[displaymath]{lineno}
\newcommand*\patchAmsMathEnvironmentForLineno[1]{%
  \expandafter\let\csname old#1\expandafter\endcsname\csname #1\endcsname
  \expandafter\let\csname oldend#1\expandafter\endcsname\csname end#1\endcsname
  \renewenvironment{#1}%
     {\linenomath\csname old#1\endcsname}%
     {\csname oldend#1\endcsname\endlinenomath}}%
\newcommand*\patchBothAmsMathEnvironmentsForLineno[1]{%
  \patchAmsMathEnvironmentForLineno{#1}%
  \patchAmsMathEnvironmentForLineno{#1*}}%
\AtBeginDocument{%
\patchBothAmsMathEnvironmentsForLineno{equation}%
\patchBothAmsMathEnvironmentsForLineno{align}%
\patchBothAmsMathEnvironmentsForLineno{flalign}%
\patchBothAmsMathEnvironmentsForLineno{alignat}%
\patchBothAmsMathEnvironmentsForLineno{gather}%
\patchBothAmsMathEnvironmentsForLineno{multline}%
}

\journal{Information Processing Letters}

\begin{document}

\begin{frontmatter}



\title{Polynomial Kernels for Tracking Shortest Paths\tnoteref{t1}}


\author{V\'aclav Bla\v zej}
\ead{vaclav.blazej@fit.cvut.cz}

\author{Pratibha Choudhary}
\ead{pratibha.choudhary@fit.cvut.cz}

\author{Du{\v s}an Knop}
\ead{dusan.knop@fit.cvut.cz}

 \author{Jan Maty{\'a}{\v s} K{\v r}i{\v s}{\v t}an\corref{cor1}}
\ead{kristja6@fit.cvut.cz}

\author{Ond\v rej~Such\'{y}}
\ead{ondrej.suchy@fit.cvut.cz}

\author{Tom{\'a}{\v s} Valla}
\ead{tomas.valla@fit.cvut.cz}

\cortext[cor1]{Corresponding author; kristja6@fit.cvut.cz}

\address{%
    Department of Theoretical Computer Science, Faculty of Information Technology, \\
    Czech Technical University in Prague, Th{\'a}kurova 9, Prague, 160\,00, Czech Republic%
}


\begin{abstract}
 Given an undirected graph $G=(V,E)$, vertices $s,t\in V$, and an integer $k$, \tsp requires deciding whether there exists a set of $k$ vertices $T\subseteq V$ such that for any two distinct shortest paths between $s$ and $t$, say $P_1$ and $P_2$, we have $T\cap V(P_1)\neq T\cap V(P_2)$.
 In this paper, we give the first polynomial size kernel for the problem.
 Specifically we  show the existence of a kernel with $\mathcal{O}(k^2)$ vertices and edges in general graphs and a kernel with $\mathcal{O}(k)$ vertices and edges in planar graphs for the \tpdag problem.
This problem admits a polynomial parameter transformation to \tsp, and this implies a kernel with $\mathcal{O}(k^4)$ vertices and edges for \tsp in general graphs and a kernel with $\mathcal{O}(k^2)$ vertices and edges in planar graphs.
Based on the above we also give a single exponential algorithm for \tsp in planar graphs.
\end{abstract}




\end{frontmatter}


\section{Introduction}\label{sec:introduction}

Given a graph $G$ with a source vertex $s$, a destination vertex $t$, and a budget~$k$, the \tsp problem is defined as follows.
Find a set of vertices $T \subseteq V(G)$ (\emph{tracking set}) with at most~$k$ vertices such that the set of vertices of $T$ encountered in each distinct shortest path between $s$ and $t$ (\emph{\stpath}) contains is a unique subset of $T$.
The purpose of finding a tracking set is to be able to distinguish between any two \sstpaths in the graph by using (preferably small number of) vertices in those paths.
The problem is known to be NP-complete~\cite{BanikKPS20} and FPT when parameterized by the size of the solution \cite{caldam18,tr1-j}.

A closely related problem is that of \tp, where the aim is to find a tracking set that distinguishes between all \stpaths and not just the \sstpaths.
The \textsc{Tracking Paths} problem admits a quadratic kernel for general graphs~\citep{quad} and a linear kernel for planar graphs~\citep{ep-planar} when parameterized by the solution size i.e., the size of tracking set.
However, so far no polynomial kernels have been known for \tsp.

In this paper, we fill this gap by computing a kernel with $\mathcal{O}(k^2)$ vertices and edges for \tpdag and a kernel with $\mathcal{O}(k)$ vertices and edges for \tpdag in planar graphs.
The kernel of size $O(k)$ leads to the first single exponential algorithm known for the \tsp in planar graphs.
It is known that \tsp is reducible to \tpdag~\cite{tr1-j}.
We show that \tpdag admits a polynomial parameter transformation to \tsp (see Section~\ref{sec:ppt}), which implies a kernel for for \tsp with size square of the size of that for \tpdag.

\paragraph{Preliminaries}
The \tsp problem is formally defined as follows.

\decprob{\tsp}{An undirected graph $G=(V,E)$, source $s \in V$, destination $t\in V$, and a budget $k \in \N$.}{Is there a set of vertices $T \subseteq V$, with $|T| \le k$, such that for any two distinct \sstpaths $P_1$ and $P_2$ in $G$, $T \cap V(P_1) \neq T \cap V(P_2)$?}

The set of vertices comprising the output is referred to as a \textit{tracking set}, and the vertices in a \textit{tracking set} are called \textit{trackers}.
We consider the parameterized version of the problem with respect to the desired solution size $k$.

We consider only unweighted graphs.
We use $G=(V,E)$ to denote a graph with vertex set $V$ and edge set $E$.
The neighborhood of a vertex $v$ comprises the vertices adjacent to $v$ and is denoted by $N(v)$.
The degree of a vertex $v$ is denoted by $\deg(v)$ and is equal to the size of the neighborhood of the vertex~$v$.
In the case of directed graphs, for a vertex $v$, $N^+(v)$ denotes the set of out-neighbors of $v$, $N^-(v)$ denotes the in-neighbors of $v$, and $\deg^+(v)=|N^+(v)|$, $\deg^-(v)=|N^-(v)|$.
A directed acyclic graph is a directed graph without any directed cycles.
For a set of edges $\{(u_1,v_1),\dots,(u_\ell,v_\ell)\}$ we call the vertices $\{v_1,\dots,v_\ell\}$ heads and $\{u_1,\dots,u_\ell\}$ tails.
For more details on graph theory and related notations, we refer to~\citet{diestel} and for details on parameterized complexity, we refer to~\citet{book}.

\section{Common Preprocessing}

In this section, we discuss some preliminary preprocessing that is performed on the input graph.
For both of our kernelization algorithms, we start by exhaustively applying the following reduction rule.

\begin{Reduction Rule}[\citet{tr1-j}]\label{red:useful-dag}
 If there exists a vertex or edge in $G$ that does not participate in any shortest \stpath, delete it.
\end{Reduction Rule}

Next, we reduce the input instance to an instance of \tpdag{} by directing the edges away from $s$.
Note that edges that have both endpoints in the same distance from $s$ do not participate in any shortest \stpath.
Furthermore, any shortest \stpath{} becomes a directed \stpath.

We proceed with applying two more reduction rules which are applicable on directed acyclic graphs.

\begin{Reduction Rule}[{\citet{tr1-j}}]\label{red:deg-one-st}
 If $\deg(s)=1$ and $u\in N^+(s)$, then delete~$s$ and set $s=u$.
 If $\deg(t)=1$ and $v\in N^-(t)$, then delete $t$ and set $t=v$.
\end{Reduction Rule}

\begin{Reduction Rule}[{\citet{tr1-j}}]\label{red:deg-2-dag}
 If there exist $x,y,z\in V(G)$, and $(x,y),(y,z)\in E(G)$, and $\deg(x)=\deg(y)=2$, then delete the vertex $y$ and introduce the edge $(x,z)$ in $G$.
\end{Reduction Rule}

After applying Reduction Rules~\ref{red:useful-dag}--\ref{red:deg-2-dag} all vertices and edges of the input graph appear in some \stpath, there are no degree-one vertices, and there are no adjacent degree-two vertices, unless one of them is $s$ or $t$.

\section{Quadratic Kernel for Directed Acyclic Graphs}\label{sec:kernel}

In this section, we give an algorithm to compute an $\mathcal{O}(k^2)$ kernel for \tpdag.
Let $(G=(V,E),s,t,k)$ be an instance of \tpdag reduced with respect to Reduction Rules~\ref{red:useful-dag}--\ref{red:deg-2-dag}.

Let $T$ be a tracking set for $G$.
For a vertex $x \in V \setminus \{t\}$ let $Z_x$ be the set of vertices $z$ such that there is a directed path from $x$ to $z$ in $G \setminus \big((T \cup \{t\})\setminus \{x\}\big)$.
Note that $x \in Z_x$.
Let $B_x$ be the set of vertices $b$ of $T \cup \{t\}$ such that there is a directed path from $x$ to $b$ in $G \setminus \big((T \cup \{t\}) \setminus \{x,b\}\big)$.

\begin{lemma}\label{lem:below_bound}
 If $T$ is a tracking set, then for every $x \in V \setminus \{t\}$ we have, \\ $\sum_{z \in Z_x} (\deg^+(z) - 1) \le |B_x|$.
\end{lemma}

\begin{proof}
 Let $E_x$ be the set of edges with tails in $Z_x$.
 Then observe that $|E_x| = \sum_{z \in Z_x} \deg^+(z)$.
 \begin{claim}
  If there are two edges in $E_x$ with the same vertex as the head, then $T$ is not tracking.
 \end{claim}

 \begin{proof}
  Assume that there are two edges $(z,y)$ and $(z',y)$ in $E_x$.
  By the definition of $Z_x$ there is a directed path $P_z$ from $x$ to $z$ that does not contain any vertex of $(T \cup \{t\})\setminus \{x\}$ and a directed path $P_{z'}$ from $x$ to $z'$ also does not contain any vertex of $(T \cup \{t\})\setminus \{x\}$.
  Due to application of Reduction Rule~\ref{red:useful-dag} each vertex and edge participate in an \stpath.
  Therefore, there is a directed path $P_y$ from $y$ to $t$ and a directed path $P_x$ from $s$ to $x$.
  Let $P_1$ be the \stpath obtained as a concatenation of $P_x$, $P_z$, $(z,y)$, $P_y$ and $P_2$ be the \stpath obtained as a concatenation of $P_x$, $P_{z'}$, $(z',y)$, $P_y$.
  These two paths differ only in the subpaths $P_z$, $(z,y)$ and $P_{z'}$, $(z',y)$ and no internal vertex of these subpaths is in $T$.
  Hence, $P_1$ and $P_2$ are two distinct paths that contain the same set of trackers and $T$ is not tracking.
 \end{proof}

 \begin{claim}
  If $y$ is the head of an edge of $E_x$ that is not in $Z_x$, then it is in $B_x$.
 \end{claim}

 \begin{proof}
  Suppose that $(z,y) \in E_x$ and $y \notin Z_x \cup B_x$.
  According to the definition of $Z_x$, there is a directed path $P_z$ from $x$ to $z$ that does not contain any vertex of $(T \cup \{t\})\setminus \{x\}$.
  Let $P_y$ be the path obtained by concatenating $P_z$ with $(z,y)$.
  If $y \notin T \cup \{t\}$, then $P_y$ is a directed path from $x$ to $y$ that does not contain any vertex of $(T \cup \{t\})\setminus \{x\}$, contradicting $y \notin Z_x$.
  Hence, $y \in T \cup \{t\}$ and $P_y$ prove that $y\in B_x$.
 \end{proof}

 By the above claims and since $x \in Z_x$ is not a head of any edge in $E_x$, we have $\sum_{z \in Z_x} \deg^+(z) = |E_x| \le |Z_x| -1 + |B_x|$ 
 and the lemma follows.
\end{proof}

\begin{lemma}\label{lem:sum_outdegree}
 If $T$ is a tracking set, then $\sum_{z \in V} (\deg^+(z)-1) \le {(|T|+1)}^2 - 1$.
\end{lemma}

\begin{proof}
 We claim that $V \setminus \{t\} = \bigcup_{x \in T \cup \{s\}} Z_x$.
 For a vertex $z \in V \setminus \{t\}$, let $P_z$ be an \stpath containing $z$.
 Let $x$ be the last vertex of $T \cup \{s\}$ on $P_z$.
 Then $z \in Z_x$.
 Hence,
 \begin{align*}
  \sum_{z \in V} \left(\deg^+(z)-1\right) &\le \deg^+(t)-1 +\!\!\!\sum_{x \in T \cup \{s\}}\sum_{\,z\in Z_x}\left(\deg^+(z)-1\right) \\
                                          &\putabove{(a)}{\le} -1 + \;\smashoperator{\sum_{x \in T \cup \{s\}}} |B_x| \\
                                          &\putabove{(b)}{\le} -1 + \;\smashoperator{\sum_{x \in T \cup \{s\}}}\big(|T|+1\big) \\
                                          &\le {(|T|+1)}^2 - 1.
 \end{align*}
 Inequality~\mbox{(a)} comes from Lemma~\ref{lem:below_bound} and Inequality~\mbox{(b)} comes from the fact that $B_x \subseteq T \cup \{t\}$.
\end{proof}

\begin{lemma}\label{lem:final-quad}
 If $(G=(V,E),s,t,k)$ is a YES instance of \tpdag after applying the Reduction Rules~\ref{red:useful-dag}--\ref{red:deg-2-dag}, then $|V|\leq 5{(k+1)}^2$ and $|E| \leq 6{(k+1)}^2$.
\end{lemma}

\begin{proof}
 Let $T$ be a tracking set of size at most $k$ in $G$.
 Let $V_2$ be the set of vertices with out-degree at least $2$.
 Note that all other vertices have out-degree at least $1$, except for $t$, which has out-degree $0$.
 We have
 \begin{align*}
  |V_2| &= \sum_{z \in V_2} 1 \\
        &\le \sum_{z \in V_2} (\deg^+(z)-1) \\
        &\putabove{(a)}{\le} 1+\sum_{z \in V} (\deg^+(z)-1) \\
        &\putabove{(b)}{\le} 1+{(|T|+1)}^2-1 \\
        &= {(|T|+1)}^2.
 \end{align*}
 Inequality~\mbox{(a)} comes from $V_2 \subseteq V$ and adds $1$ for vertex $t$, \mbox{(b)} comes from Lemma~\ref{lem:sum_outdegree}.
 Therefore, there are at most $\sum_{z \in V_2} (\deg^+(z)) = |V_2| +\sum_{z \in V_2} (\deg^+(z)-1) \le |V_2|+1+\sum_{z \in V} (\deg^+(z)-1) \le 2{(|T|+1)}^2$ edges with tails in vertices with out-degree at least $2$.
 There are at most $|V_2| \le {(|T|+1)}^2$ edges with heads in vertices with out-degree at least $2$ and in-degree $1$.
 By a symmetrical argument, there are at most $2{(|T|+1)}^2$ edges with heads in vertices with in-degree at least $2$ and at most ${(|T|+1)}^2$ edges with tails in vertices with in-degree at least $2$ and out-degree $1$.

 As the graph is reduced by Reduction Rules~\ref{red:useful-dag}--\ref{red:deg-2-dag} every edge is incident to a vertex with in-degree at least $2$ or out-degree at least $2$.
 Hence, there are at most $6{(|T|+1)}^2 \le 6{(k+1)}^2$ edges in total.
 Each vertex with in-degree $1$ and out-degree $1$ is incident to an edge with tail in vertex with total degree at least $3$.
 Therefore, there are at most $3{(|T|+1)}^2$ vertices with in-degree $1$ and out-degree $1$ and, hence, at most $5{(|T|+1)}^2 \le 5{(k+1)}^2$ vertices in total.
\end{proof}

Thus, in the case of a YES instance, the number of vertices and edges in the graph can be bounded by $\mathcal{O}(k^2)$, where $k$ is the desired size of a tracking set.
After application of Reduction Rules~\ref{red:useful-dag}--\ref{red:deg-2-dag} if there are more than $5{(k+1)}^2$ vertices or more than $6{(k+1)}^2$ edges in the graph, then we return a NO instance, otherwise we have a kernel.
Since all reduction rules are applicable in polynomial time, our kernelization algorithm runs in polynomial time.

\begin{theorem}\label{thm:tsp_kernel_regular}
 \tpdag admits a kernel with $\mathcal{O}(k^2)$ vertices and edges, where $k$ is the size of a solution.
\end{theorem}

\section{Linear Kernel for Planar Directed Acyclic Graphs}\label{sec:linear-kernel}


In this section, we give a linear kernel when the input is a planar graph.
The algorithm uses a set of reductions and a lemma of \citet[Lemma~4.6]{BanikKPS20}, which shows bounds for a slightly different setting.
Once again, we start by reducing the problem to an instance, say $(G,s,t,k)$, of \tpdag and applying Reduction Rules~\ref{red:useful-dag}--\ref{red:deg-2-dag}.


We use the term \textit{reduced graph} to refer to the graph on which Reduction Rule~\ref{red:useful-dag} has already been applied exhaustively.

\begin{lemma}[{\citet[Lemma~4.6]{BanikKPS20}}]\label{lem:bn-planar}
 The number of faces $|F|$ in a planar embedding of a reduced planar graph is at most $2\cdot {\rm OPT}$, where ${\rm OPT}$ is the number of trackers in an optimum tracking set for the graph.
\end{lemma}

We would like to note that although the above lemma was given for undirected graphs, it can be observed that orientation of edges does not affect the lemma statement.
Indeed, while reducing an instance of \tsp to \tpdag, the number of faces in the graph remains unchanged.

It follows from Lemma~\ref{lem:bn-planar} that if the number of faces in an embedding of a reduced planar graph exceeds $2k$, then the graph does not have a tracking set of size $k$.
Therefore, we have the following reduction rule.

\begin{Reduction Rule}\label{red:bound-faces}
 If a planar embedding of a reduced planar graph $G$ has more than $2k$ faces, then return a trivial NO instance.
\end{Reduction Rule}

Next, we introduce two more reduction rules.
These rules are simply directed variants of two rules given by~\citet{ep-planar}; hence the proofs of their correctness are omitted to avoid redundancy.

\begin{Reduction Rule}\label{red:dag-triangle}
 If there exists a vertex $v\in V(G)\setminus \{s,t\}$ with degree $2$ with incident edges $(u,v), (v,w)$, and $(u,w)\in E(G)$, then mark $v$ as a tracker and remove it.
 Set $k=k-1$.
\end{Reduction Rule}

\begin{Reduction Rule}\label{red:disjoint-paths-dag}
 If there exist two vertices $x$ and $y$ with degree $2$, adjacent to the same pair of vertices $u$ and $v$, i.e., $(u,x),(u,y),(x,v)$, and $(y,v)\in E(G)$, then mark $x$ as a tracker and delete it.
 Set $k=k-1$.
\end{Reduction Rule}

In the case of Reduction Rule~\ref{red:dag-triangle}, there exists an \stpath that contains the vertex $v$, say $P_v$, and so there also exists another \stpath with the same vertex set as that of $P_v$, except that edges $(u,v),(v,w)$ are replaced by the edge $(u,w)$. Hence, $v$ needs to be marked a tracker.
In the case of Reduction Rule~\ref{red:disjoint-paths-dag}, there is an \stpath containing vertex $x$, say $P_x$, and so there is another \stpath which has the same vertex set as $P_x$, except that the vertex $x$ is replaced by vertex $y$. Hence, either $x$ or $y$ necessarily needs to be marked as a tracker.

Now we prove the existence of a linear kernel.

\begin{lemma}\label{thm:planar-kernel}
 For a planar graph $G=(V,E)$ if $(G,s,t,k)$ is a YES instance of \tpdag reduced with respect to Reduction Rules~\ref{red:useful-dag}--\ref{red:disjoint-paths-dag}, then $|V| \leq 10k-10$ and $|E|\leq 12k-12$.
\end{lemma}

\begin{proof}
 Let $G$ be a reduced planar DAG and $F$ be the set of faces in some planar drawing of $G$.
 Note that, since the instance is reduced with respect to Reduction Rule~\ref{red:bound-faces}, we have $|F| \le 2k$.
 Let $V_{\geq 3}$ be the set of vertices with degree at least $3$ and $V_2$ be the set of vertices with degree equal to $2$.
 Recall that after exhaustive application of Reduction Rules~\ref{red:useful-dag}--\ref{red:deg-2-dag} there are no vertices with degree $1$ in $G$ and each vertex with degree $2$ has vertices with degree $3$ or more as its neighbors.

 First we construct an auxiliary graph $G'=(V',E')$ by short-circuiting all vertices in $V_2$ (the vertex is deleted and an edge is introduced between its neighbors).
 Due to Reduction Rules~\ref{red:dag-triangle} and \ref{red:disjoint-paths-dag}, the short-circuiting does not create parallel edges.
 Note that the number of faces $|F|$ in $G$ is the same as that in $G'$, that is, at most $2k$.
 Further, $|V'|=|V_{\geq 3}|$, $|V_2|\leq |E'|$, and $|E|\leq 2|E'|$.

 We now bound the size of $V_{\geq 3}$.
 Due to the Handshaking lemma and the fact that degree of all vertices in $G'$ is at least $3$, we get
 \begin{align}\label{eq:1}
  |E'|\geq 3|V'|/2.
 \end{align}
 Due to Euler's formula, we have
 \noindent
 \begin{equation*}
  |F| = |E'|-|V'|+2 \putabove{(\ref{eq:1})}{\geq} |V'|/2 + 2.
 \end{equation*}
 Hence,
 \noindent
 \[
  |V_{\geq 3}|=|V'| \leq 2(|F|-2) \putabove{Red.~\ref{red:bound-faces}}{\leq} 2(2k - 2) \leq 4k - 4.
 \]

 Thus, $|V_2|\leq |E'|=|F|+|V'|-2 \leq 2k + 4k-4-2=6k-6$.
 The total number of vertices in $G$ is $|V|=|V_2|+|V_{\geq 3}|\leq 6k-6+4k-4 = 10k-10$.
 Since $|E|\leq 2|E'|$,  $|E|\leq 12k-12$.
\end{proof}

Thus, after the application of Reduction Rules~\ref{red:useful-dag}--\ref{red:disjoint-paths-dag}, if there exist more than $10k-10$ vertices or $12k-12$ edges, we return a NO instance; otherwise, we have a kernel.
Since all reduction rules are applicable in polynomial time, our kernelization algorithm runs in polynomial time, therefore, Theorem~\ref{thm:tsp_kernel_planar} follows.

\begin{theorem}\label{thm:tsp_kernel_planar}
 \tpdag in planar graphs admits a kernel with $\mathcal{O}(k)$ vertices and edges.
\end{theorem}

Once a kernel is computed, a tracking set of size $k$ can be found by considering all subsets of vertices of size $k$, and verifying if any of them is a tracking set for the given graph using the polynomial time algorithm of Banik et al.~\cite{tr-j}.

Since we have at most $10k$ vertices in the kernel, we have $\sum_{i=1}^{k}\binom{10k}{i}$ subsets of size at most $k$.
We use \cite[Lemma~3.13]{exact_exp_algo} which states that
\[
 \sum_{i=1}^{\alpha n}\binom ni \le \left(\frac 1\alpha \right)^{\alpha n} \!\!\cdot \left( \frac 1{1-\alpha} \right)^{(1-\alpha)n}.
\]
for $\alpha \in (0,1)$.
By setting $\alpha = 1/10$ and $n=10k$ we get the following bound.
\[
 \sum_{i=1}^k\binom {10k}i
 \le \left(\frac 1{\frac 1{10}} \right)^k \cdot \left( \frac 1{1-\frac 1{10}} \right)^{(1-\frac 1{10})10k}
 \!\!= \left(\frac {{10}^{10}}{9^9} \right) ^k \approx 25.811747^k \le 26^k
\]
Thus, a tracking set can be found in $26^k\cdot n^{\mathcal{O}(1)}$ time in this case.
Note that the tracking set will also include  vertices that were marked as trackers and deleted in the intermediate steps of the algorithm.

\begin{corollary}\label{cor:single-exp}
 There exists a $26^k \cdot n^{\mathcal{O}(1)}$ time (FPT) algorithm for \tsp and \tpdag in planar graphs.
\end{corollary}

\section{Polynomial Parameter Transformation}\label{sec:ppt}
It has been shown that \tsp is reducible to \tpdag~\citep{tr1-j}.
In this section, we show how to reduce \tpdag to \tsp.

\begin{lemma}\label{lem:subdivide}
 Let $(G,s,t,k)$ be an instance of\/ \tpdag, where $G$ does not contain parallel edges.
 Let $(u,w)$ be an edge of $G$, and let $G'$ be obtained from $G$ by subdividing the edge $(u,w)$, that is, by removing the edge $(u,w)$, introducing a new vertex $v$, and adding edges $(u,v)$ and $(v,w)$.
 Then $(G',s,t,k)$ is a YES instance of\/ \tpdag if and only if $(G,s,t,k)$ is a YES instance of\/ \tpdag.
\end{lemma}

\begin{proof}
 There is a one-to-one correspondence between the directed \stpaths in $G$ and the directed \stpaths in $G'$.
 Hence, if $T$ is a solution for $(G,s,t,k)$, then it is also a solution for $(G',s,t,k)$.
 Similarly, if $T$ is a solution for $(G',s,t,k)$ that does not contain $v$, then it is also a solution for  $(G,s,t,k)$.
 Hence, assume that there is a solution $T$ for $(G',s,t,k)$ which contains $v$.
 Our aim is to show that there is also a solution $T'$ for $(G',s,t,k)$ which does not contain $v$ and hence is also a solution for $(G,s,t,k)$.

 Suppose first that there is a directed $u$-$w$ path $Q$ in $G'$ which does not contain vertices of $T$ as internal vertices.
 Note that, as $v \in T$, this path does not contain~$v$.
 If $Q$ is the single edge $(u,w)$, then there are parallel edges in $G$ that contradict our assumptions.
 Hence, $Q$ contains at least one internal vertex.
 Let $v'$ be an arbitrary internal vertex of $Q$.
 We claim that $T'=(T \setminus \{v\}) \cup \{v'\}$ is a solution for $(G',s,t,k)$.

 \begin{figure}[h]
  \centering
  \includegraphics[scale=1.2,page=1]{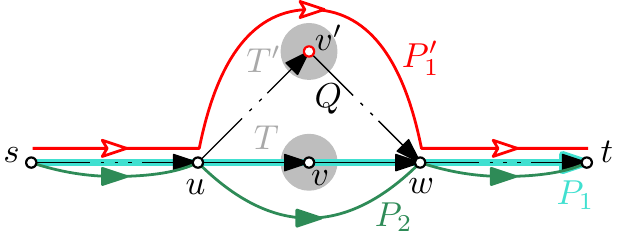}
  \caption{
   The case where there exists a directed $u$-$w$ path in $G'$ which does not contain vertices of $T$ as internal vertices.
  }%
  \label{fig:edge_subdivision_proof_case1}
 \end{figure}

 Assume it is not and let $P_1$ and $P_2$ be two different \stpaths such that $V(P_1) \cap T' = V(P_2) \cap T'$, see Figure~\ref{fig:edge_subdivision_proof_case1}.
 Since $V(P_1) \cap T \neq V(P_2) \cap T$, exactly one of the paths contains~$v$.
 Without loss of generality, assume that $v \in V(P_1)$ and $v \notin V(P_2)$.
 Note that $v' \notin V(P_1)$ since the graph does not have directed cycles.
 Hence, $v' \notin V(P_2)$ as $V(P_1) \cap T' = V(P_2) \cap T'$.
 Therefore $V(P_2) \cap T = V(P_2) \cap T'$.
 Now consider the path $P'_1$ obtained from $P_1$ by replacing the subpath $u,v,w$ with $Q$.
 Since $Q$ does not contain vertices of $T$ as internal vertices, we have $V(P'_1) \cap T = V(P_1) \cap T'$.
 It follows that 
 \[
  V(P'_1) \cap T = V(P_1) \cap T' = V(P_2) \cap T' = V(P_2) \cap T
 \]
 which contradicts that $T$ is a solution for $(G',s,t,k)$.

 Now assume that there is no directed $u$-$w$ path in $G'$ that does not contain vertices of $T$ as internal vertices, see Figure~\ref{fig:edge_subdivision_proof_case2}.
 Consider the set $T_u = (T \setminus \{v\}) \cup \{u\}$ and the set $T_w = (T \setminus \{v\}) \cup \{w\}$.
 We will show that if neither $T_u$ nor $T_v$ is a solution in $(G',s,t,k)$, then $T$ is not a solution in $(G',s,t,k)$, a contradiction.

 \begin{figure}[h]
  \centering
  \includegraphics[width=.98\textwidth,page=2]{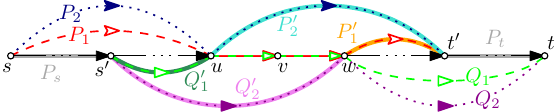}
  \caption{
   The case where there is no path from $u$ to $w$ which does not contain vertices of $T$.
   We show that a path through $Q'_2$ and $P'_1$ has the same trackers in $T$ as the one through $Q'_1$ and $P'_2$.
  }%
  \label{fig:edge_subdivision_proof_case2}
 \end{figure}

 If $T_u$ is not a solution, then there are two \stpaths $P_1$ and $P_2$ such that $V(P_1) \cap T_u = V(P_2) \cap T_u$.
 Since $V(P_1) \cap T \neq V(P_2) \cap T$, exactly one of the paths, say $P_1$, contains $v$.
 This implies $u \in V(P_1)$ and, as $u \in T_u$, also $u \in V(P_2)$.
 Let $t'$ be the first common vertex of the paths $P_1$ and $P_2$ after $u$.
 Let $P'_1$ be the part of the path $P_1$ between $w$ and $t'$, $P'_2$ be the part of the path $P_2$ between $u$ and $t'$ (note the slight asymmetry).
 Since $V(P_1) \cap T_u = V(P_2) \cap T_u$, neither $P'_1$ nor $P'_2$ contains vertices of $T$ as internal vertices.
 By assumption, this implies $t' \neq w$, as otherwise $P'_2$ would be a $u$-$w$ path.
 This, in turn, implies $w \notin T$.

 Now, the same argument could be repeated on a graph with swapped labels $u$ and $w$, $s$ and $t$, and where direction of all edges is reversed.
 We present it for completeness.
 If $T_w$ is not a solution, then there are two \stpaths $Q_1$ and $Q_2$ such that $V(Q_1) \cap T_w = V(Q_2) \cap T_w$.
 Since $V(Q_1) \cap T \neq V(Q_2) \cap T$, exactly one of the paths, say $Q_1$, contains $v$.
 This implies $w \in V(Q_1)$ and, as $w \in T_w$, also $w \in V(Q_2)$.
 Let $s'$ be the last common vertex of the paths $Q_1$ and $Q_2$ before $w$.
 Let $Q'_1$ be the part of the path $Q_1$ between $s'$ and $u$, $Q'_2$ be the part of the path $Q_2$ between $s'$ and $w$ (again the slight asymmetry).
 Since $V(Q_1) \cap T_w = V(Q_2) \cap T_w$, neither $Q'_1$ nor $Q'_2$ contains vertices of $T$ as internal vertices.
 By assumption, this implies $s' \neq u$.
 This, in turn, implies $u \notin T$.

 Let $P_s$ be any path from $s$ to $s'$, e.g., a part of $Q_1$ and
 let $P_t$ be any path from $t'$ to $t$, e.g., a part of $P_1$.
 Let $R_1$ be the \stpath formed by the concatenation of $P_s$, $Q'_1$, $P'_2$, and $P_t$;
 also let $R_2$ be formed by the concatenation of $P_s$, $Q'_2$, $P'_1$, and $P_t$.
 Note that $R_1$ goes through $u$ and $R_2$ goes through $w$, but neither of them contains $v$.
 These paths only differ in the part between $s'$ and $t'$.
 Namely, $R_1$ contains $Q'_1$ and $P'_2$ while $R_2$ contains $Q'_2$ and $P'_1$.
 As argued above, neither of $P'_1,P'_2,Q'_1,Q'_2$ contains a vertex of $T$ as an internal vertex.
 Moreover, neither $u$ nor $w$ is in $T$.
 It follows that $V(R_1) \cap T=V(R_2) \cap T$ which contradicts that $T$ is a solution for $(G',s,t,k)$.
 Therefore, $T_u$ or $T_w$ is a solution for $(G',s,t,k)$ which finishes the proof.
\end{proof}

\begin{lemma}\label{lem:tpdag-to-tsp}
 There is a polynomial parameter transformation from \tpdag{} to \tsp{} mapping instances $(G,s,t,k)$ to instances $(G',s,t,k)$ such that $|V(G')|, |E(G')|\leq |E(G)|\cdot|V(G)|$.
 %
\end{lemma}

\begin{proof}
Let $(G=(V,E),s,t,k)$ be an instance of \tpdag.
We first compute the prospective layer $L(v)$ for each vertex $v$ by the following procedure.
Start by assigning $L(s) =0$ for $s$ (and all other sources if there are any; there are none in our kernels).
For all other vertices $v$ in topological order, assign $L(v)= 1+ \max \{L(u) \mid u \in N^-(v)\}$.
Note that, as we always only add $1$ for each vertex, we have $\max \{L(v) \mid v \in V\} \le |V|-1$.

Next, we create an auxiliary directed graph $G''$ from $G$ as follows.
For each edge $(u,v)$ with $L(u)+1 < L(v)$ we subdivide the edge $L(v)-L(u)-1$ times to obtain a path with $L(v)-L(u)-1$ internal vertices.
In $G''$ every directed \stpath{} has length exactly $L(t)$.
Also it follows from Lemma~\ref{lem:subdivide} that $(G'',s,t,k)$ is a YES instance of \tpdag if and only if $(G,s,t,k)$ is a YES instance of \tpdag.

Let $G'$ be the underlying undirected graph of $G''$.
Every \stpath{} in $G'$ is of length at least $L(t)$ and all paths of length exactly $L(t)$ are actually directed paths in $G''$, since the edges of $G''$ only connect consecutive layers.
Hence, the shortest \stpaths{} in $G'$ are exactly the directed \stpaths{} in $G''$.


Hence,  $(G',s,t,k)$ is a YES instance of \tsp if and only if $(G,s,t,k)$ is a YES instance of \tpdag.
Observe that for each edge of $G$ the corresponding path in $G'$ contains at most $|V|-2$ new vertices and edges.
Since each path can be of length at most $V(G)$, the number of vertices and edges in $G'$ is at most $|E(G)|\cdot|V(G)|$.
\end{proof}

\begin{theorem}
 \tsp admits a kernel with $\mathcal{O}(k^4)$ vertices and edges, where $k$ is the size of a solution.
 \tsp in planar graphs admits a kernel with $\mathcal{O}(k^2)$ vertices and edges.
\end{theorem}
\begin{proof}
 Lemma~\ref{lem:tpdag-to-tsp} implies that if there exists a kernel of size $X$ for an instance of \tpdag, then there exists a kernel of size $X^2$ for the corresponding instance of \tsp.
 Theorem~\ref{thm:tsp_kernel_regular} gives a kernel for \tpdag with $\mathcal{O}(k^2)$ vertices and edges, and Theorem~\ref{thm:tsp_kernel_planar} gives a kernel with $\mathcal{O}(k)$ vertices and edges in planar graphs.
 Applying the implied bound to these kernels for \tpdag yields the desired kernels for \tsp.
\end{proof}

\section*{Declaration of competing interest}
The authors declare that they have no known competing financial interests or personal relationships that could have appeared to influence the work reported in this paper.

\section*{Acknowledgements}
This work was supported by the Grant Agency of the Czech Technical University in Prague, grant \mbox{No.~SGS20/208/OHK3/3T/18}. The authors acknowledge the support of the OP VVV MEYS funded project CZ.02.1.01/0.0/0.0/16\_019/0000765 ``Research Center for Informatics''.




\bibliographystyle{elsarticle-num-names}


\bibliography{ref}

%
%
%
\end{document}